\documentclass[11pt,a4paper]{article}
\usepackage{jheppub}
\usepackage{amsmath}
\usepackage[most]{tcolorbox}
\usepackage{dsfont}
\usepackage{ulem}
\usepackage{natbib}
\usepackage{xcolor}
\usepackage[hang,flushmargin]{footmisc}
\usepackage{tikz-cd}
\usepackage{enumitem}
\usepackage{amsfonts,amssymb, amscd,amsmath,latexsym,amsbsy,bm,amsthm}
\usepackage{stmaryrd}
\usepackage{todonotes}
\usepackage{float}
\usepackage{tikz}
\usetikzlibrary{positioning}

\usepackage{nameref}

\newtheorem{theorem}{Theorem}[section]
\newtheorem{lemma}[theorem]{Lemma}
\newtheorem{definition}{Definition}[section]

\usepackage{romannum}

\makeatletter\renewcommand{\@biblabel}[1]{#1.}\makeatother

\newtcolorbox{empheqboxed}{colback=gray!20, 
 colframe=white,
 width=\textwidth,
 sharpish corners,
 top=0mm, 
 bottom=0pt
}

\title{Decorating the gauge/YBE correspondence
}
\author{Erdal Catak$^a$ and Mustafa Mullahasanoglu$^{b}$}
\affiliation{
$^{a}$ Department of Physics, Istanbul University,\\ 34134 Istanbul, Turkey \\[-0.4cm]

$^b$ Department of Physics, Bogazici University,\\ 34342 Bebek, Istanbul, Turkey\\[-0.4cm]

}

\emailAdd{ecatak@istanbul.edu.tr}
\emailAdd{mustafa.mullahasanoglu@std.bogazici.edu.tr}

\abstract{In this paper, we aim to study the three-dimensional $\mathcal N=2$ supersymmetric dual gauge theories on $S_b^3/\mathbb{Z}_r$ in the context of the gauge/YBE correspondence. We consider hyperbolic hypergeometric integral identities acquired via the equality of supersymmetric lens partition functions as solutions to the decoration transformation and the flipping relation in statistical mechanics. 
The solutions of those transformations aim at investigating various decorated lattice models possessing the Boltzmann weights of integrable Ising-like models obtained via the gauge/YBE correspondence.
We also constructed The Bailey pairs for the decoration transformation and the flipping relation. 
}


\begin{document}

\maketitle

\section{Introduction}

Exactly solved models are a remarkable field of research in statistical mechanics. Developments in the area of exactly solved models bring new techniques and understandings both in physics and mathematics. The most known equation, actually the integrability condition for the systems, is the Yang-Baxter equation (YBE). String theory gives many novel integrable models with correspondence such as gauge/Bethe ansatz between integrable spin chains to 2d supersymmetric gauge theories \cite{Nekrasov:2009uh, Nekrasov:2009ui} and also extended to four-dimensional
gauge theories\footnote{The relation between spin chains and four-dimensional $\mathcal{N}=2$ theories see also \cite{Pomoni:2019oib}.} in \cite{Nekrasov:2009rc}, four-dimensional Chern-Simons theory for generating integrable
$\sigma$-models \cite{Costello:2017dso, Costello:2018gyb, Costello:2019tri}
and there are also plenty of results for classical Yang-Baxter equations \cite{Catal-Ozer:2019tmm, Bakhmatov:2022rjn, Gubarev:2023jtp} from string theory. 

The Yang-Baxter equation has been also investigated under the gauge/YBE correspondence \cite{Spiridonov:2010em, Yamazaki:2012cp}, see also reviews \cite{Gahramanov:2017ysd, Yamazaki:2018xbx, Yagi:2016oum} allowing us to find new solutions to the star-triangle relation, the simplest version of YBE for edge interacting lattice spin models, from the equality of partition functions of the supersymmetric gauge theories.  
The star-triangle relation is the fundamental element for the exact solution in the transfer matrix approach for the square lattice model and duality transformations between the triangle and hexagon lattice models in two-dimensional planar space. Another integrability condition for the edge-interacting lattice spin models is the star-star relation \cite{Baxter:1997ssr} which can be derived by the existence of the star-triangle relation. 
The YBE for vertex type (spins interact over vertices) and interaction round a face (IRF) type (spins interact through a face) models are also studied by constructing Bailey pairs \cite{Gahramanov:2015cva} and deriving IRF-type YBE \cite{Baxter:1997ssr} with the help of the star-triangle relation or the star-star relation. 
Up to now, the integrability of lattice spin models by the star-triangle and the star-star relation and their non-planar duals consisting of higher-spin interactions \cite{Mullahasanoglu:2023nes} by the star-square relation and the generalized star-triangle relation are investigated in the context of the gauge/YBE correspondence.

However, adding secondary spins between each nearest-neighbor interaction is also possible by decoration transformation \cite{Naya1954, Fisher1959}. The transformation provides the construction of bond-decorated lattice spin models like a Kagome-lattice model \cite{Naya1954} and helps solve it by the corresponding integrable model. Ising-like integrable models are the models on two-dimensional planar space covered with one-unit elements. However, one can also have the lattice model consisting of different types of cells and the most known model is the Kagome-lattice model. It is important that it has the same coordination number with a square lattice but its critical temperature is not the same but lower. The decoration transformation is also a tool to solve various bond-decorated models \cite{Syozi:1980iw, ROJAS20091419, Stre_ka_2010} as well.

In this work, we obtain solutions to the decoration transformation and the flipping relation by considering some three-dimensional $\mathcal N=2$ supersymmetric dualities on $S^3_b/\mathbb{Z}_r$. Both transformations are solved for the two integrable lattice spin models which consist of different gauge symmetries $SU(2)$ case \cite{Gahramanov:2016ilb} and $U(1)$ case \cite{Bozkurt:2020gyy} in the context of the supersymmetric gauge theories. It is also studied that a solution with $U(1)$ gauge symmetry can be obtained by breaking $SU(2)$ the gauge symmetry. 
We applied the gauge symmetry-breaking method \cite{Spiridonov:2010em} that we reduce $SU(2)$ gauge symmetries to $U(1)$ gauge symmetries in three-dimensional $\mathcal N=2$ supersymmetric gauge theories. From a statistical mechanics point of view, it is interesting that the gauge symmetry breaking reduces Boltzmann weights in the solutions of the decoration transformation and the flipping relation as studied for the star-triangle relation \cite{Bozkurt:2020gyy} and the star-star relation \cite{Mullahasanoglu:2021xyf}.


We also construct Bailey pairs for the decoration transformation and the flipping relation. We realize that one of the Bailey pairs can be kept the same, but it is translated by its parameters by a trivial operator. The other partner of the Bailey pairs can be translated non-trivially through a Bailey chain of infinite length. We investigate this operation on the Bailey pairs which are recently studied in \cite{Gahramanov:2022jxz} for hyperbolic hypergeometric integral identities. However, the investigation of the Coxeter relations \cite{Gahramanov:2015cva} and the reduced R-operators in vertex-type YBE are left for further studies. We complete the work by introducing the decoration transformation and the flipping relation for IRF-type models \cite{BAXTER1986321, Date:1987wy, date:1987-2}.

The rest of the paper is organized as follows.
In section 2, we introduce the solutions of the decoration transformation and the flipping relation by the use of integral identities coming from the equality of the partition functions of $\mathcal{N}=2$ supersymmetric gauge theories on $S^3_b/\mathbb{Z}_r$. In section 3, we construct Bailey pairs for both relations considered in the previous section. In section 5, we derive both relations for IRF-type models as well, and the derivations are depicted. Finally, in section 6, we conclude the results and discuss further studies. The Appendix \ref{appendix1} consists of notations and mathematical tools.

\section{Three-dimensional IR dualities}

Historically, Seiberg duality is introduced for the four-dimensional $\mathcal N=1$ supersymmetric gauge theories. Here, we will discuss it as certain three-dimensional $\mathcal N=2$ supersymmetric gauge theories \cite{Intriligator:1996ex, Aharony:1997bx}. Dual theories have the following idea: The physical observables are the same at an infrared fixed point but different at the UV level. The equality of partition functions of supersymmetric theories studied on $S^3$ \cite{Kapustin:2010xq}, $S^3_b$  \cite{Dolan:2011rp, Gahramanov:gka, Amariti:2015vwa}, and $S^3_b/\mathbb{Z}_r$ \cite{Benini:2011nc, Imamura:2012rq, Imamura:2013qxa} and the equality of the  superconformal indices \cite{Krattenthaler:2011da, Kapustin:2011jm, Gahramanov:2013rda, Gahramanov:2016wxi} are some of the evidences for Seiberg dualities.

Here, we particularly study the three-dimensional $\mathcal N=2$ partition functions on the squashed lens space $S^3_b/\mathbb{Z}_r$.
The computation of partition functions can be done in several ways such as dimensional reduction of the four-dimensional lens superconformal index \cite{Benini:2011nc, Yamazaki:2013fva, Nieri:2015yia, Eren:2019ibl} and the supersymmetric localization technique \cite{Imamura:2012rq, Imamura:2013qxa}.

\subsection{The decoration transformation}

We consider decoration -or iteration- transformation \cite{Naya1954, Fisher1959} for lattice spin models in statistical mechanics.
The decoration transformation is a map between a spin system consisting of two outer spins interacting with a central spin and a two-spin system with a single interaction. It is a tool to acquire solutions for decorated models since it relates the partition functions of the integrable model and its decorated version up to some coefficient. In this study, we will work on the solutions to the decoration transformation to decorate integrable lattice spin models obtained via the gauge/YBE correspondence. 

In \cite{Kels:2018xge}, the decoration transformation\footnote{The decoration transformation is also called "chain relation" in the discussion of Zamolodchikov's 
fishnet model \cite{ZAMOLODCHIKOV198063, Kazakov:1983ns, Derkachov:2023xqq}.} is discussed as a triangle identity that appears in the algebraic structure of lattice spin models. The transformation is depicted in Fig.\ref{decorationfigure} and the mathematical expression is the following
\begin{align}
   \sum_{m_0} \int dx_0\: S(\sigma_0) W_{\alpha}(\sigma_1,\sigma_0)W_{\beta}(\sigma_2,\sigma_0)
   =\mathcal{R}(\alpha,\beta) W_{\alpha + \beta}(\sigma_1,\sigma_2)
   \label{decorationdefinition},
\end{align}
where $\mathcal{R}(\alpha,\beta)$ and $S(\sigma_0)$ stand for a spin-independent function and self-interaction term, respectively, and $\alpha, \beta$ are the spectral parameters without the condition crossing parameter\footnote{In the solutions of the star-triangle relation, we discuss the crossing parameter \cite{Gahramanov:2016ilb}.}. 
Integration for continuous spin variables $x_0$ and summation for discrete spin variables $m_0$ in the left-hand side of (\ref{decorationdefinition}) are evaluated over the center spin $\sigma_0=(x_0,m_0)$ as shown its elimination in  Fig.\ref{decorationfigure}. 

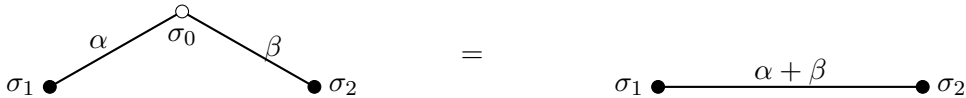
\begin{figure}[tbh]
\centering
\begin{tikzpicture}[scale=2]

\draw[-,thick] (-2,0)--(-2.87,-0.5);
\draw[-,thick] (-2,0)--(-1.13,-0.5);

\filldraw[fill=white,draw=black] (-2,0) circle (1.2pt)
node[below=2.5pt]{\color{black} $\sigma_0$};
\filldraw[fill=black,draw=black] (-2.87,-0.5) circle (1.2pt)
node[left=1.5pt] {\color{black} $\sigma_1$};
\filldraw[fill=black,draw=black] (-1.13,-0.5) circle (1.2pt)
node[right=1.5pt] {\color{black} $\sigma_2$};

\filldraw[fill=white,draw=black] (-2.55,-0.05) 
node[below=2.5pt]{\color{black} $\alpha$};
\filldraw[fill=white,draw=black] (-1.4,-0.05) 
node[below=2.5pt]{\color{black} $\beta$};

\fill[white!] (0.05,-0.3) circle (0.01pt)
node[left=0.05pt] {\color{black}$=$};

\draw[-,thick] (1.13,-0.5)--(2.87,-0.5);

\filldraw[fill=black,draw=black] (1.13,-0.5) circle (1.2pt)
node[left=1.5pt]{\color{black} $\sigma_1$};
\filldraw[fill=black,draw=black] (2.87,-0.5) circle (1.2pt)
node[right=1.5pt]{\color{black} $\sigma_2$};

\filldraw[fill=white,draw=black] (2,-0.2) 
node[below=2.5pt]{\color{black} $\alpha+\beta$};

\end{tikzpicture}
\caption{The decoration transformation.}
\label{decorationfigure}
\end{figure}

We note that the decoration transformation can be obtained by reducing spins from the star-triangle relation and the reduction is the decrease of the number of flavors from the supersymmetric gauge theories aspect. 
 
\subsubsection{ \texorpdfstring{$SU(2)$}{SU(2)} gauge theory}

Here, we study the equality of partition functions of three-dimensional $\mathcal N=2$ theories on the squashed lens space $S^3_b/\mathbb{Z}_r$. The equality of partition functions with $N_f=6$ flavors obtained via dimensional reduction \cite{Benini:2011nc, Yamazaki:2013fva, Eren:2019ibl} or localization techniques \cite{Imamura:2012rq, Imamura:2013qxa} is interpreted as star-triangle relation in \cite{Sarkissian:2018ppc, Gahramanov:2016ilb}. 

The ingredients of the dual theories can be described as the following: The electric theory has $SU(2)$ gauge symmetry and $N_f=6$ flavors and the magnetic theory consists of fifteen chiral multiples and does not have gauge degrees of freedom. The chiral multiplets of the electric theory transform under the fundamental representation of the gauge group and the flavor group and the vector multiple of the same theory transforms as the adjoint representation of the gauge group while the chiral multiples of the magnetic theory are in the totally antisymmetric tensor representation of the flavor group. 

Here, we investigate the lens hyperbolic hypergeometric\footnote{The lens hyperbolic hypergeometric gamma function and the notations are introduced in Appendix \ref{appendix1}.} integral identity\footnote{The beta integrals have the balancing condition on the parameters $a_i$ but there is no balancing condition in this study.} for the $N_f=4$ case which is obtained in \cite{Mullahasanoglu:2021xyf} by reducing the number of flavors
\begin{align} \nonumber
    &\frac{1}{2r\sqrt{-\omega_1\omega_2}} \sum_{y=-[ r/2 ]}^{[ r/2 ]} \int _{-\infty}^{\infty} \frac{\prod_{i=1}^4\gamma_h(a_i\pm z,u_i\pm y;\omega_1,\omega_2)}{\gamma_h(\pm 2z,\pm 2y;\omega_1,\omega_2)}  dz
     \nonumber \\
     &=
          \gamma_h\left(\omega_1+\omega_2-\sum_{i=1}^4a_i,-\sum_{i=1}^4u_i;\omega_1,\omega_2\right)
          \prod_{1\leq i<j\leq 4}\gamma_h(a_i + a_j,u_i + u_j;\omega_1,\omega_2)  \; ,
          \label{decorationequation}
\end{align}
where the summation is over the holonomies $y=\frac{r}{2\pi} \int A_\mu d x^{\mu}$, where $A_\mu$ is the gauge field, and the integration is performed over a non-trivial cycle on $S_b^3/{\mathbb Z}_r$.
It is interesting that the integral identity (\ref{decorationequation}) is studied as a solution to the star-square relation \cite{Mullahasanoglu:2023nes} and its $r=1$ case is written as the star-triangle relation in \cite{Kels:2018xge}. 

To obtain the solution of the decoration transformation, one can change the variables as 
\begin{align}
    	\begin{aligned}
	a_{1,2} & =-\alpha_1\pm x_{1}\;, ~~~~~~~~~~~~~\; a_{3,4}=-\alpha_2\pm x_{2}  \;, \\
	u_{1,2} & =-\beta_1\pm m_{1}\;, ~~~~~~~~~~~~~\; v_{3,4}=-\beta_2\pm m_{2}\:,
	\end{aligned}
	\end{align}
and the Boltzmann weights becomes
\begin{align}
\begin{aligned} \label{boltamannweightSU2}
    W_{\alpha_i,\beta_i}(x_i,x_j,m_i,m_j) &=    
    \gamma_h(-\alpha_i\pm x_i\pm x_j,-\beta_i\pm m_i\pm  m_j;\omega_{1},\omega_2)    \;,
\end{aligned}
	\end{align}
where $\omega_1, \omega_2$ are introduced as temperature-like parameters. The Boltzmann weight (\ref{boltamannweightSU2}) also solves the star-triangle relation \cite{Gahramanov:2016ilb} and the star-star relation \cite{Mullahasanoglu:2021xyf} of the corresponding integrable model. There is also the spin-independent function which can be normalized but we left it to be in the same notation as previous works 
\begin{align}
    \begin{aligned}
R(\alpha_1,\alpha_2,\beta_1,\beta_2)=&
      \gamma_h\left(\omega +2(\alpha_1+\alpha_2) ,2(\beta_1+\beta_2); \omega_1 , \omega_2\right)
      \\ &\times
\prod_{i=1}^2\gamma_h(-2\alpha_i ,-2\beta_i; \omega_1 , \omega_2)\:.
\label{spindep}
    \end{aligned}
\end{align}
and the self-interaction term due to the external field
\begin{align}
    S(x_0,m_0)=\frac{1}{\gamma_h(\pm 2x_0,\pm 2m_0;\omega_1,\omega_2)}\:.
\end{align}

After the change of variables, the integral identity (\ref{decorationequation}) takes the following form
	\begin{align}
     \begin{aligned}
	\sum_{m_0=-[ r/2 ]}^{\lfloor r/2\rfloor} \int_{-\infty}^\infty S(x_0,m_0)&\prod_{i=1}^2 W_{\alpha_i,\beta_i}(x_{i},x_0,m_i,m_0)
 \frac{dx_0}{r\sqrt{-\omega_{1}\omega_{2}}}\\ &=R(\alpha_1,\alpha_2,\beta_1,\beta_2) \: W_{\alpha_1+\alpha_2,\beta_1+\beta_2}(x_{1},x_{2},m_1,m_2) \;,
\label{decoratiosolution}
 \end{aligned}
	\end{align}
where it is shown that the equality of partition functions of the supersymmetric gauge theories can be studied as the solution of the decoration transformation.

\subsubsection{ \texorpdfstring{$U(1)$}{U(1)} gauge theory}

We construct the decoration transformation for the generalized Faddeev-Volkov model. From a supersymmetric gauge theory perspective, the integral identity which will be investigated as the decoration transformation is the equality of three-dimensional $\mathcal{N}=2$ supersymmetric gauge theories on $S_b^3/{\mathbb Z}_r$ and one of the dual theories has the $U(1)$ gauge symmetry. The equality is obtained by breaking the gauge symmetry from $SU(2)$ to $U(1)$ but the integral identity is first discussed in \cite{Bag:2022kla} as a flavor limit of the duality of $\mathcal{N}=2$ supersymmetric gauge theories. 

In this $U(1)$ gauge duality, one of the dual theories has the $U(1)$ gauge symmetry and $SU(2)_L \times SU(2)_R$ flavor symmetry without the Chern-Simons term. In this theory, two chiral multiplets belong to the fundamental representation but the other two belong to the anti-fundamental representation of the flavor group. However, its dual theory has no gauge symmetry but the same global symmetry where those four chiral multiplets transform in the fundamental representation of the flavor group $SU(2)_L \times SU(2)_R$.

From a statistical mechanics point of view, the certain case ($r=1$) of the equality of partition functions of the gauge theories is also considered as the solution to the star-triangle relation in \cite{Bazhanov:2015gra, Kels:2018xge}. 

However, here, we show that the following integral identity can be written as the solution to the decoration transformation \cite{Naya1954, Fisher1959, Syozi:1980iw} as well, 
 \begin{align}
\begin{aligned}
\label{decorationequation2}
    &\sum_{y=-[ r/2 ]}^{[ r/2 ]}  \int _{-\infty}^{\infty} \prod_{i=1}^2  e^{i\pi \big(\frac{za_i}{\omega_1 \omega_2r} 
    + \frac{yu_i}{r}\big) } \gamma_h(a_i-z,u_i- y;\omega_1,\omega_2)
    \\ & \quad \quad \quad\quad\quad\quad\quad
    \times  
    e^{i\pi \big(\frac{zb_i}{\omega_1 \omega_2r} 
    + \frac{yv_i}{r}\big) }
    \gamma_h(b_i+ z,v_i+y;\omega_1,\omega_2)
\frac{dz}{r\sqrt{-\omega_1\omega_2}}
 \\
&\quad\quad=\gamma_h\left(\omega -\sum_{i=1}^2 (a_i+b_i),r-\sum_{i=1}^2 (u_i+v_i); \omega_1 , \omega_2\right)   
   \\ &\quad \quad\quad\times 
    \prod_{i,j=1}^2 \gamma_h(a_i+ b_j,u_i+ v_j;\omega_1,\omega_2)
    \prod_{i,j=1,\:i\neq j}^2 e^{\frac{i\pi}{2} \big(\frac{(a_i+b_i)(a_j-b_j)}{\omega_1 \omega_2 r} + \frac{(u_i+v_i)(u_j-v_j)}{r}  \big)}\;. 
\end{aligned}
\end{align}	

We apply the following change of variables 
\begin{align}
    	\begin{aligned}
	a_i & =-\alpha_i+x_{i}\;, ~~~~~~~~~~~~~\; b_{i}=-\alpha_i-x_{i}  \;, \\
	u_i & =-\beta_i+y_{i}\;, ~~~~~~~~~~~~~\; v_{i}=-\beta_i-y_{i}\:.
	\end{aligned}
	\end{align}
then, the Boltzmann weights becomes
\begin{align}
\begin{aligned} \label{boltamannweightU1}
    W_{\alpha_i,\beta_i}(x_i,x_j,y_i,y_j) &=    
    \gamma_h(-\alpha_i+x_i- x_j,-\beta_i+y_i- y_j;\omega_{1},\omega_2) \\&\times 
    \gamma_h(-\alpha_i-x_i+ x_j,-\beta_i-y_i+ y_j;\omega_{1},\omega_2)   \;,
\end{aligned}
	\end{align}
and the spin-independent function is the same as (\ref{spindep}). The Boltzmann weight (\ref{boltamannweightU1}) also solves the star-triangle relation \cite{Bozkurt:2020gyy} and the star-star relation \cite{Catak:2021coz} of the generalized Faddeev-Volkov model.
 
The self-interactions for secondary spins (unfilled spin in Fig.\ref{decorationfigure}) at the decorated lattice model is
\begin{align}
   S(x_0,y_0)= e^{-2i\pi \big(\frac{x_0(\alpha_1+\alpha_2)}{\omega_1 \omega_2r} 
        + \frac{y_0(\beta_1+\beta_2)}{r}\big) }  \:,
\end{align}
and the self-interactions for the spins at the right-hand side of the Fig.\ref{decorationfigure}
\begin{align}
   S(x_1,y_1)=e^{-2i\pi \big(\frac{x_1\alpha_2}{\omega_1 \omega_2r} 
        + \frac{y_1\beta_2}{r}\big) }  \:, \quad\quad   S(x_2,y_2)=e^{-2i\pi \big(\frac{x_2\alpha_1}{\omega_1 \omega_2r} 
        + \frac{y_2\beta_1}{r}\big) }  \:.
\end{align}
If one substitutes the definition of Boltzmann weights in the integral identity (\ref{decorationequation2}), one obtains the decoration transformation for the generalized Faddeev-Volkov model.
\subsubsection{Gauge symmetry breaking}\label{gsb1}

We apply the gauge symmetry breaking method \cite{Spiridonov:2010em, Bozkurt:2020gyy} to acquire the integral identity (\ref{decorationequation2}) which is previously obtained by reducing the number of flavors of the dual supersymmetric gauge theories in \cite{Bag:2022kla}. The equality of partition functions of supersymmetric theories with $SU(2)$ gauge symmetry (\ref{decorationequation}) is reduced to the integral identity (\ref{decorationequation2}) standing for the duality of the supersymmetric theories with $U(1)$ gauge symmetry. Hence the idea of the gauge symmetry breaking method is the reduction of the gauge symmetry from the $SU(2)$ gauge group to the $U(1)$ gauge group. It is also important to remark that the flavor symmetry group is also broken from $SU(4)$ to $SU(2)\times SU(2)$ during the gauge symmetry breaking.

We also apply the procedure to make it evident that the application of the gauge symmetry breaking allows us to reach another solution of the relations for lattice spin models in statistical mechanics studied such as the star-triangle relation \cite{Spiridonov:2010em, Bozkurt:2020gyy} and the star-star relation \cite{Mullahasanoglu:2021xyf}. The Boltzmann weights discussed above as the solution to the decoration transformation do not have a direct reduction from (\ref{boltamannweightSU2}) to (\ref{boltamannweightU1}). Therefore, it is interesting that the gauge symmetry-breaking method plays a direct role between relations but not the Boltzmann weights of the lattice spin models.

The gauge symmetry-breaking protocol will be applied step-by-step as the following 

\begin{itemize}
    \item The boundaries of the integrals will be changed from $\{-\infty,\infty \}$ to $\{0,\infty\}$ due to the symmetry $z \leftrightarrow -z$ and the integration variable is translated $z\to z+\mu$.
    \item The flavor fugacities are reparametrized as $a_i$ with $a_i+\mu$ for $i=1,2$ and $a_i-\mu$ for  $i=3,4$ \:.
    \item Then the limit $\mu\to \infty$ is taken and the asymptotic properties of the hyperbolic hypergeometric functions (\ref{asymptotics}) are used.
    \item Finally, the exponents are organized and some fugacities are relabeled as $a_{3,4}=b_{1,2}$ and $u_{3,4}=v_{1,2}$\:.
\end{itemize}

\subsubsection{Kagome lattice model}

The decoration transformation with the star-triangle relation is firstly used to solve the Kagome lattice model \cite{Naya1954}. The importance of the Kagome model is that it has the same number of interactions for each site as a square lattice but they have different critical points. It is just a certain example of many other extended lattices by the use of symmetry transformations \cite{Fisher1959, Syozi:1980iw}.

The Kagome lattice as depicted in Fig.\ref{kagome} can be obtained by the use of the star-triangle relation and the decoration transformation. One can start with the honeycomb lattice and apply decoration transformation. At each bond of the honeycomb lattice, there will appear secondary spin presented in the middle of Fig.\ref{kagome}. Then the central spins (spins of the honeycomb lattice model) will be integrated by the use of the star-triangle relation. Finally, we reach the Kagome lattice in which there are four edge interactions for each spin.


\begin{figure}[tbh]
\centering
\begin{tikzpicture}[scale=0.5]

\begin{scope}[xshift=300pt]

\draw[-,very thick] (-1.35,-2.45)--(0,-1.75);
\draw[-,very thick] (1.35,-2.45)--(0,-1.75);

\draw[-,very thick] (-1.35,2.45)--(0,1.75);
\draw[-,very thick] (1.35,2.45)--(0,1.75);

\draw[-,very thick] (-1.52,0.85)--(0,1.75);
\draw[-,very thick] (-1.52,-0.85)--(0,-1.75);
\draw[-,very thick] (-1.52,-0.85)--(-1.35,-2.45);
\draw[-,very thick] (-1.52,0.85)--(-1.35,2.45);
\draw[-,very thick] (-1.52,-0.85)--(-2.65,0);
\draw[-,very thick] (-1.52,0.85)--(-2.65,0);
\draw[-,very thick] (-1.52,0.85)--(-1.52,-0.85);

\draw[-,very thick] (1.52,0.85)--(0,1.75);
\draw[-,very thick] (1.52,-0.85)--(0,-1.75);
\draw[-,very thick] (1.52,-0.85)--(1.35,-2.45);
\draw[-,very thick] (1.52,0.85)--(1.35,2.45);
\draw[-,very thick] (1.52,-0.85)--(2.65,0);
\draw[-,very thick] (1.52,0.85)--(2.65,0);
\draw[-,very thick] (1.52,0.85)--(1.52,-0.85);

\filldraw[fill=white,draw=black] (-1.35,-2.45) circle (4.0pt)
node[below=1pt]{\small };
\filldraw[fill=white,draw=black] (1.35,-2.45) circle (4.0pt)
node[below=1pt]{\small };
\filldraw[fill=white,draw=black] (2.65,0) circle (4.0pt)
node[right=1pt]{\small };
\filldraw[fill=white,draw=black] (-1.35,2.45) circle (4.0pt)
node[above=1pt]{\small };
\filldraw[fill=white,draw=black] (1.35,2.45) circle (4.0pt)
node[above=1pt]{\small };
\filldraw[fill=white,draw=black] (-2.65,0) circle (4.0pt)
node[left=1pt]{\small };

\filldraw[fill=white,draw=black] (-1.52,-0.85) circle (4.0pt)
node[below=1pt]{\small };
\filldraw[fill=white,draw=black] (-1.52,0.85) circle (4.0pt)
node[below=1pt]{\small };
\filldraw[fill=white,draw=black] (1.52,-0.85) circle (4.0pt)
node[right=1pt]{\small };
\filldraw[fill=white,draw=black] (1.52,0.85) circle (4.0pt)
node[above=1pt]{\small };
\filldraw[fill=white,draw=black] (0,-1.75) circle (4.0pt)
node[above=1pt]{\small };
\filldraw[fill=white,draw=black] (0,1.75) circle (4.0pt)
node[left=1pt]{\small };

\filldraw[fill=black,draw=black] (-15,0)
node[left=1pt]{= };

\end{scope}

\begin{scope}[xshift=0pt]

\draw[-,very thick] (2,0)--(1,1.73);
\draw[-,very thick] (1,1.73)--(1.70,3.10);
\draw[-,very thick] (1,-1.73)--(2,0);
\draw[-,very thick] (1,-1.73)--(1.70,-3.10);
\draw[-,very thick] (3.30,0)--(2,0);
\draw[-,very thick] (-1,1.73)--(-2,0);
\draw[-,very thick] (-1,-1.73)--(-2,0);
\draw[-,very thick] (-1,-1.73)--(-1.70,-3.10);
\draw[-,very thick] (-1,1.73)--(-1.70,3.10);
\draw[-,very thick] (-3.30,0)--(-2,0);
\draw[-,very thick] (-1,-1.73)--(1,-1.73);
\draw[-,very thick] (-1,1.74)--(1,1.74);

\filldraw[fill=black,draw=black] (-1.7,-3.1) circle (2.0pt)
node[below=1pt]{\small $\sigma_E$};
\filldraw[fill=black,draw=black] (1.7,-3.1) circle (2.0pt)
node[below=1pt]{\small $\sigma_D$};
\filldraw[fill=black,draw=black] (3.3,0) circle (2.0pt)
node[right=1pt]{\small $\sigma_C$};
\filldraw[fill=black,draw=black] (1.7,3.1) circle (2.0pt)
node[above=1pt]{\small $\sigma_B$};
\filldraw[fill=black,draw=black] (-1.7,3.10) circle (2.0pt)
node[above=1pt]{\small $\sigma_A$};
\filldraw[fill=black,draw=black] (-3.3,0) circle (2.0pt)
node[left=1pt]{\small $\sigma_F$};

\filldraw[fill=white,draw=black] (-1.35,-2.45) circle (4.0pt)
node[below=1pt]{\small };
\filldraw[fill=white,draw=black] (1.35,-2.45) circle (4.0pt)
node[below=1pt]{\small };
\filldraw[fill=white,draw=black] (2.65,0) circle (4.0pt)
node[right=1pt]{\small };
\filldraw[fill=white,draw=black] (-1.35,2.45) circle (4.0pt)
node[above=1pt]{\small };
\filldraw[fill=white,draw=black] (1.35,2.45) circle (4.0pt)
node[above=1pt]{\small };
\filldraw[fill=white,draw=black] (-2.65,0) circle (4.0pt)
node[left=1pt]{\small };

\filldraw[fill=white,draw=black] (-1.52,-0.85) circle (4.0pt)
node[below=1pt]{\small };
\filldraw[fill=white,draw=black] (-1.52,0.85) circle (4.0pt)
node[below=1pt]{\small };
\filldraw[fill=white,draw=black] (1.52,-0.85) circle (4.0pt)
node[right=1pt]{\small };
\filldraw[fill=white,draw=black] (1.52,0.85) circle (4.0pt)
node[above=1pt]{\small };
\filldraw[fill=white,draw=black] (0,-1.75) circle (4.0pt)
node[above=1pt]{\small };
\filldraw[fill=white,draw=black] (0,1.75) circle (4.0pt)
node[left=1pt]{\small };

\filldraw[fill=black,draw=black] (-1,-1.73) circle (2.0pt)
node[below=1pt]{\small };
\filldraw[fill=black,draw=black] (1,-1.73) circle (2.0pt)
node[below=1pt]{\small };
\filldraw[fill=black,draw=black] (2,0) circle (2.0pt)
node[right=1pt]{\small };
\filldraw[fill=black,draw=black] (1,1.73) circle (2.0pt)
node[above=1pt]{\small };
\filldraw[fill=black,draw=black] (-1,1.73) circle (2.0pt)
node[above=1pt]{\small };
\filldraw[fill=black,draw=black] (-2,0) circle (2.0pt)
node[left=1pt]{\small };
\filldraw[fill=black,draw=black] (6,0)
node[left=1pt]{= };

\end{scope}

\begin{scope}[xshift=-300pt]

\draw[-,very thick] (2,0)--(1,1.73);
\draw[-,very thick] (1,1.73)--(1.70,3.10);
\draw[-,very thick] (1,-1.73)--(2,0);
\draw[-,very thick] (1,-1.73)--(1.70,-3.10);
\draw[-,very thick] (3.30,0)--(2,0);
\draw[-,very thick] (-1,1.73)--(-2,0);
\draw[-,very thick] (-1,-1.73)--(-2,0);
\draw[-,very thick] (-1,-1.73)--(-1.70,-3.10);
\draw[-,very thick] (-1,1.73)--(-1.70,3.10);
\draw[-,very thick] (-3.30,0)--(-2,0);
\draw[-,very thick] (-1,-1.73)--(1,-1.73);
\draw[-,very thick] (-1,1.74)--(1,1.74);

\filldraw[fill=black,draw=black] (-1.7,-3.1) circle (2.0pt)
node[below=1pt]{\small $\sigma_E$};
\filldraw[fill=black,draw=black] (1.7,-3.1) circle (2.0pt)
node[below=1pt]{\small $\sigma_D$};
\filldraw[fill=black,draw=black] (3.3,0) circle (2.0pt)
node[right=1pt]{\small $\sigma_C$};
\filldraw[fill=black,draw=black] (1.7,3.1) circle (2.0pt)
node[above=1pt]{\small $\sigma_B$};
\filldraw[fill=black,draw=black] (-1.7,3.10) circle (2.0pt)
node[above=1pt]{\small $\sigma_A$};
\filldraw[fill=black,draw=black] (-3.3,0) circle (2.0pt)
node[left=1pt]{\small $\sigma_F$};

\filldraw[fill=black,draw=black] (-1,-1.73) circle (2.0pt)
node[below=1pt]{\small };
\filldraw[fill=black,draw=black] (1,-1.73) circle (2.0pt)
node[below=1pt]{\small };
\filldraw[fill=black,draw=black] (2,0) circle (2.0pt)
node[right=1pt]{\small };
\filldraw[fill=black,draw=black] (1,1.73) circle (2.0pt)
node[above=1pt]{\small };
\filldraw[fill=black,draw=black] (-1,1.73) circle (2.0pt)
node[above=1pt]{\small };
\filldraw[fill=black,draw=black] (-2,0) circle (2.0pt)
node[left=1pt]{\small };

\end{scope}

\end{tikzpicture}
\caption{Constructing Kagome lattice from the hexagonal lattice by the use of decoration transformation and the star-triangle relation, respectively.}
\label{kagome}
\end{figure}
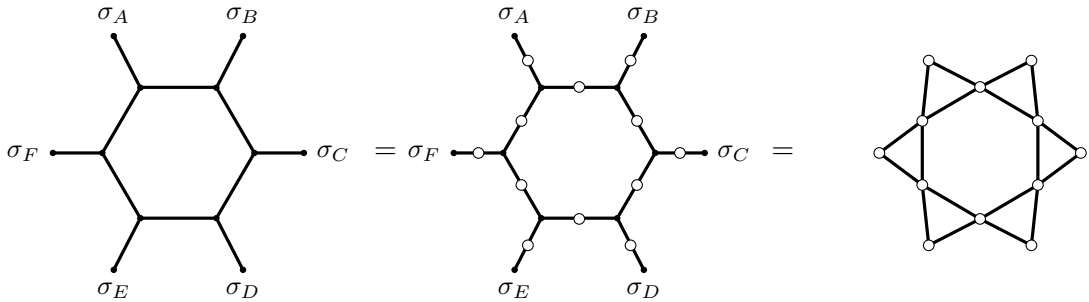

The investigation of the associated quiver diagram of the Kagome lattice is also similar and interesting. The details of the derivation are left to further studies. 
\subsection{The flipping relation}
From an integrability perspective, we will study the reduced version of the star-star relation. The reduction can be seen as certain Boltzmann weights in the star-star relation are eliminated or taken trivially. We name\footnote{We also note similar discussions and possible relations with the fields the 2-2 Pachner moves \cite{PACHNER1991129, Dimofte:2011ju, Nagao:2011aa, Terashima_2014} and the supersymmetric gauge theories \cite{Cheng_2023, Bajeot:2023gyl, Razamat:2019sea, Aprile:2018oau}.} it as a "flipping relation", depicted in Fig.\ref{flippingfigure} since two spins standing opposite sites interacting with the centered spin exchange the interactions after making use of the relation. 


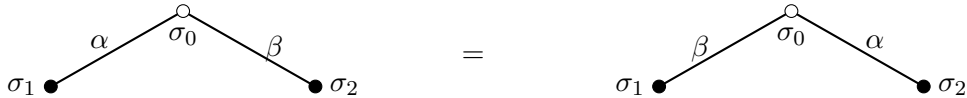
\begin{figure}[tbh]
\centering
\begin{tikzpicture}[scale=2]

\draw[-,thick] (-2,0)--(-2.87,-0.5);
\draw[-,thick] (-2,0)--(-1.13,-0.5);

\filldraw[fill=white,draw=black] (-2,0) circle (1.2pt)
node[below=2.5pt]{\color{black} $\sigma_0$};
\filldraw[fill=black,draw=black] (-2.87,-0.5) circle (1.2pt)
node[left=1.5pt] {\color{black} $\sigma_1$};
\filldraw[fill=black,draw=black] (-1.13,-0.5) circle (1.2pt)
node[right=1.5pt] {\color{black} $\sigma_2$};

\filldraw[fill=white,draw=black] (-2.55,-0.05) 
node[below=2.5pt]{\color{black} $\alpha$};
\filldraw[fill=white,draw=black] (-1.4,-0.05) 
node[below=2.5pt]{\color{black} $\beta$};

\fill[white!] (0.05,-0.3) circle (0.01pt)
node[left=0.05pt] {\color{black}$=$};

\draw[-,thick] (2,0)--(1.13,-0.5);
\draw[-,thick] (2,0)--(2.87,-0.5);

\filldraw[fill=white,draw=black] (2,0) circle (1.2pt)
node[below=2.5pt]{\color{black} $\sigma_0$};
\filldraw[fill=black,draw=black] (1.13,-0.5) circle (1.2pt)
node[left=1.5pt] {\color{black} $\sigma_1$};
\filldraw[fill=black,draw=black] (2.87,-0.5) circle (1.2pt)
node[right=1.5pt] {\color{black} $\sigma_2$};

\filldraw[fill=white,draw=black] (2.55,-0.05) 
node[below=2.5pt]{\color{black} $\alpha$};
\filldraw[fill=white,draw=black] (1.4,-0.05) 
node[below=2.5pt]{\color{black} $\beta$};

\end{tikzpicture}
\caption{The flipping relation.}
\label{flippingfigure}
\end{figure}

\begin{align}
   \sum_{m_0} \int dx_0\: S(\sigma_0) W_{\alpha}(\sigma_1,\sigma_0)W_{\beta}(\sigma_2,\sigma_0)
   = \sum_{y_0} \int dz_0\: S(\sigma_0) W_{\beta}(\sigma_1,\sigma_0)W_{\alpha}(\sigma_2,\sigma_0)
   \label{flippingdefinition},
\end{align}

The flipping relation is discussed for IRF and vertex models as a trivial solution to the star-triangle relation since it represents the self-commuting transfer matrices \cite{BAXTER1986321}. We will already discuss the derivation of the decoration and the flipping relation for IRF models but the reduction of IRF flipping relation to the edge interacting flipping relation like a reduction from IRF-YBE to the star-star relation \cite{Bazhanov1992} will be left as a further studies.


We will study the duality of supersymmetric gauge theories in which both dual theories possess $SU(2)$ gauge symmetries. The equality of partition functions is obtained by the double integral method first used in \cite{van2007hyperbolic}, see also \cite{Catak:2021coz}. The integral identity (\ref{decorationequation}) provides us the following transformation formula 

\begin{align}
      \frac{1}{2r\sqrt{-\omega_1\omega_2}} &\sum_{y=-[ r/2 ]}^{[ r/2 ]} \int _{-\infty}^{\infty} dx\frac{\prod_{i=1}^4\gamma_h(a_i\pm z,u_i\pm y;\omega_1,\omega_2)}{\gamma_h(\pm 2z,\pm 2y;\omega_1,\omega_2)}
\nonumber\\ 
&=\frac{\gamma_h(a_1+a_2,u_1+u_2;\omega_1,\omega_2)\gamma_h(a_3+a_4,u_3+u_4;\omega_1,\omega_2)}{\gamma_h(\Tilde{a}_1+\Tilde{a}_2,\Tilde{u}_1+\Tilde{u}_2;\omega_1,\omega_2)\gamma_h(\Tilde{a}_3+\Tilde{a}_4,\Tilde{u}_3+\Tilde{u}_4;\omega_1,\omega_2)}
\nonumber\\ &\times
      \frac{1}{2r\sqrt{-\omega_1\omega_2}} \sum_{m=-[ r/2 ]}^{[ r/2 ]} \int _{-\infty}^{\infty} dx\frac{\prod_{i=1}^4\gamma_h(\Tilde{a}_i\pm x,\Tilde{u}_i\pm m;\omega_1,\omega_2)}{\gamma_h(\pm 2x,\pm 2m;\omega_1,\omega_2)}
      \label{flippingequation1}
\end{align}
where we also redefine some parameters together with the modification
\begin{align}
    \Tilde{a}_{1,2}=a_{1,2}+s \:,\quad
\Tilde{a}_{3,4}=a_{3,4}-s \:,\quad
  \Tilde{u}_{1,2}= u_{1,2}+p \:,\quad
\Tilde{u}_{3,4}=u_{3,4}-p \:,
\label{tilde1}
\end{align}
with $s$ and $p$ which are chemical potentials for continuous and discrete fugacities, respectively.

The integral identity (\ref{flippingequation1}) can be studied as the flipping relation for Boltzmann weights (\ref{boltamannweightSU2}) under the proper change of the variables.

The same double integral method is also used for the Seiberg duality with $U(1)$ gauge symmetry (\ref{decorationequation2}). Then the result is the flipping relation for the Boltzmann weight (\ref{boltamannweightU1}) and appears as
\begin{align}
      \frac{1}{r\sqrt{-\omega_1\omega_2}} \sum_{y=-[ r/2 ]}^{[ r/2 ]} \int _{-\infty}^{\infty} dz\prod_{i=1}^2
      e^{\pi i\left(
\frac{z(a_i+b_i)}{\omega_{1}\omega_2r}+\frac{y(u_i+v_i)}{r} \right)}
\nonumber\\
\times \gamma_h(a_i- z,u_i-y;\omega_1,\omega_2)\gamma_h(b_i+ z,v_i+ y;\omega_1,\omega_2)
\nonumber\\ 
=\frac{e^{\frac{\pi i}{2}\left(
\frac{(a_2+b_2)(a_1-b_1)}{\omega_{1}\omega_2r}+\frac{(u_2+v_2)(u_1-v_1)}{r} \right)}}{e^{\frac{\pi i}{2}\left(
\frac{(a_1+b_1)(a_2-b_2)}{\omega_{1}\omega_2r}+\frac{(u_1+v_1)(u_2-v_2)}{r} \right)}}
\nonumber\\ \times
\frac{\gamma_h(a_1+b_1,u_1+v_1;\omega_1,\omega_2)\gamma_h(a_2+b_2,u_2+v_2;\omega_1,\omega_2)}{\gamma_h(\Tilde{a}_1+\Tilde{b}_1,\Tilde{u}_1+\Tilde{v}_1;\omega_1,\omega_2)\gamma_h(\Tilde{a}_2+\Tilde{b}_2,\Tilde{u}_2+\Tilde{v}_2;\omega_1,\omega_2)}
\nonumber\\ \times
      \frac{1}{r\sqrt{-\omega_1\omega_2}} \sum_{m=-[ r/2 ]}^{[ r/2 ]} \int _{-\infty}^{\infty} dx \prod_{i=1}^2e^{\pi i\left(
\frac{x(\Tilde{a}_i+\Tilde{b}_i)}{\omega_{1}\omega_2r}+\frac{m(\Tilde{u}_i+\Tilde{v}_i)}{r} \right)}
\nonumber\\
\times \gamma_h(\Tilde{a}_i- x,\Tilde{u}_i- m;\omega_1,\omega_2)
      \gamma_h(\Tilde{b}_i+ x,\Tilde{v}_i+m;\omega_1,\omega_2)\:,
      \label{flippingequation2}
\end{align}
where
\begin{align}
\begin{aligned}
    \tilde{a}_1= a_1+s\:,\quad
\tilde{a}_2= a_2-s\:,\quad
   \tilde{b}_1= b_1+s\:,\quad
\tilde{b}_2= b_2-s\:,
\\
    \tilde{u}_1= u_1+p\:,\quad
\tilde{u}_2= u_2-p\:,\quad
   \tilde{v}_1= v_1+p\:,\quad
\tilde{v}_2= v_2-p\:,
\end{aligned}\label{tilde2}
\end{align}
with $s\in \mathbb{C}$ and $p\in \mathbb{Z}$.

The gauge symmetry breaking is also applied to the integral identity (\ref{flippingequation1}) which is the duality of supersymmetric theories that has $SU(2)$ gauge symmetry in both theories. It is again a reduction from $SU(2)$ to $U(1)$ gauge symmetry group but for two dual theories. And also  $SU(4)$ flavor symmetry group will be reduced to $SU(2)\times SU(2)$. With some differences from Section \ref{gsb1}, the gauge symmetry-breaking procedure is similar.


\section{Bailey pairs}
\subsection{The decoration transformations}
In this section, we construct Bailey pairs for the decoration transformation. However, this construction can be seen as the reduced version of the star-triangle relation and one can note that constructing the star-triangle relation in terms of hyperbolic hypergeometric functions appeared in \cite{Gahramanov:2022jxz} seems as the generalized version of the decoration transformation in the aspect of sequences of functions. 

We will emphasize throughout the construction the details of the reduction of the star-triangle relation into the decoration transformation.

\begin{definition}
Two functions $\alpha(x,m;t,p)$ and $\beta(x,m; t,p)$, where $x, t\in \mathbb{C}$ and $m,p \in \mathbb{Z}$, form an integral hyperbolic hypergeometric Bailey pair with respect to $t$ and $p$ if the functions satisfy
\begin{equation}\label{definition}
    \beta(z,m;t,p) = M(t,p)_{z,m;x,j}\alpha(x,j;t,p) \;, 
\end{equation}
where $M(t,p)_{z,m;x,j}$ is an operator integrating over $x\in \mathbb{C}$ and summing over $j\in \mathbb{Z}$, which also called an integral-sum operator. 
\end{definition}

We also assume an identity operator $ I(s,q)$ with $s \in \mathbb{C}$ and $q\in \mathbb{Z}$, such that the operator attaches the new variables to functions in which it satisfies the relation \begin{equation}
    I(s,q)  I(-s,-q)=1\:, \quad \& \quad  I(0,0)=1\:.
\end{equation}
By the definition it also satisfies 
\begin{equation}
\begin{gathered}
       I(t,p)  I(s,q)=   I(s+t,q+p)   
      \:.
\end{gathered}
\label{homo}
\end{equation}
Suppose that the operators $M$ and $I$ satisfy the "star-triangle relation"
\begin{equation}
\label{strI}
\begin{gathered}
       M(s,q)_{w,k; z,m}  I(s+t,q+p)  M(t,p)_{z,m; x,j} = I(t,p) M(s+t,q+p)_{w,k; x,j}  I(s,q)\:,
\end{gathered}
\end{equation}
where $w \in \mathbb{C}$ and $k\in \mathbb{Z}$.

However, due to the triviality of the operator $I$ and the relation (\ref{homo}), the star-triangle relation (\ref{strI}) reduces and can be read as the "decoration transformation"
\begin{equation}
\label{operatordecoration}
\begin{gathered}
       M(s,q)_{w,k; z,m}  M(t,p)_{z,m; x,j} =  M(s+t,q+p)_{w,k; x,j}  \:.
\end{gathered}
\end{equation}

\begin{lemma}[Bailey Lemma]
When $\alpha(x,m;t,p)$ and $\beta(x,m; t,p)$ form an integral hyperbolic hypergeometric Bailey pair with respect to $t$ and $p$, the novel function as a part of the sequences of functions $\beta'(x,k; t+s,p+q)$  and reparametrized function $\alpha'(x,k;t+s,p+q)$ defined by 
\begin{align}\label{defining}
    \alpha'(x,k;t+s,p+q) &=I(s,q) \alpha(x,k; t,p) \:,\\
    \beta'(x,k; t+s,p+q)&=I(-s,-q) M(s,q)_{x,k;z,m} I(t+s,p+q)\beta(z,m; t,p)\:,
\end{align}
form an integral hyperbolic hypergeometric Bailey pair with respect to the new parameters $t+s$ and $p+q$.
\end{lemma}
  \begin{proof}
    Recall the definition of Bailey pairs (\ref{definition}) for the functions $\alpha'(x,j;t+s,p+q)$ and $\beta'(x,k; t+s,p+q)$ 
    \begin{equation}
    \beta'(w,k;t+s,p+q) = M(t+s,p+q)_{w,k;x,j}\alpha'(x,j;t+s,p+q)\:, 
\end{equation}
Substitute the functions into the relation (\ref{defining}) defining a Bailey pair 
\begin{gather}
    I(-t,-p)M(s,q)_{w,k;z,m}I(s+t,p+q) \beta(z,m; t,p) = \\
    M(s+t, p+q)_{w,k;x,j} I(s,q)\alpha(x,j; t,p) \;. 
\end{gather}
The proof is easily completed by the use of the properties of the $I$ operator and then the problem ends up with the decoration transformation
\begin{equation}
\begin{gathered}
       M(s,q)_{w,k; z,m}    M(t,p)_{z,m; x,j} =
       M(s+t,q+p)_{w,k; x,j}  \:,
\end{gathered}
\end{equation}
which is assumed to be true in (\ref{decorationdefinition}). 
 \end{proof}

\subsubsection{ \texorpdfstring{$SU(2)$}{SU(2)} gauge theory  }

In this part, we introduce the integral sum operator $M$ to construct Bailey pairs for the decoration transformation $\eqref{decorationequation}$ which is actually the equality of supersymmetric theories with $SU(2)$ gauge symmetry. The $M$ operator with non-trivial $D$ operator constructed in \cite{Gahramanov:2022jxz} satisfies the star-triangle relation. However, it will be a main part of the construction of Bailey pairs for the decoration transformation and it is
\begin{equation}
\begin{aligned}
M(t,p)_{z,m;x,j}=&
\frac{1}{C(t,p)}\sum_{j=-[ r/2 ]}^{[ r/2 ]} \int _{-\infty}^{\infty} 
\gamma_h(-t+z\pm x,m-p\pm j;\omega_1,\omega_2)      \\
	&\times  \gamma_h(-t-z\pm x,-m-p\pm j;\omega_1,\omega_2) 
	 	\frac{[\mathrm{d}_j {x}]}{2r\sqrt{-\omega_1\omega_2}}\:,
\end{aligned}\label{SU2Moperator}
\end{equation}

where the measure is
\begin{equation}
\begin{aligned}
 \relax [\mathrm{d}_j {x}] = \frac{ \mathrm{d}x}{\gamma_h(\pm 2x, \pm   2j;\omega_1,\omega_2)}\:, \label{measure}
\end{aligned}
\end{equation}

and the spin-independent function is

\begin{equation}
\begin{aligned}
C(t,p)=&\gamma_h(-2t,-2p;\omega_1,\omega_2)\;.
\label{spinindepinoperator}
\end{aligned}
\end{equation}

One can show that the operator (\ref{SU2Moperator}) satisfies the decoration transformation (\ref{operatordecoration}) by the use of integral identity (\ref{decorationequation}) with the following change of variables

\begin{equation}
\begin{aligned}
     {a}_{1,2} & =  -s\pm w, \quad {a}_{3,4}  = -t\pm x\:,
     \\
    {u}_{1,2} & =  -q\pm k, \quad {u}_{3,4}  =  -p\pm m\:.\label{changeofvariables1}
\end{aligned}
\end{equation}

More details for the calculation can be seen in Appendix A in \cite{Gahramanov:2022jxz}.

\subsubsection{ \texorpdfstring{$U(1)$}{U(1)} gauge theory }

We will construct a Bailey pair for the integral identity (\ref{decorationequation2}) where the Seiberg duality has the $U(1)$ gauge symmetry. The $M$ operator is also the same as its corresponding star-triangle relation worked in \cite{Gahramanov:2022jxz} but the $D$ operator is trivial for the decoration transformation 

\begin{equation}
\begin{aligned}
M(t,p)_{z,m;x,j}=&
\frac{1}{C(t,p)}\sum_{j=-[ r/2 ]}^{[ r/2 ]}\int _{-\infty}^{\infty} 
\gamma_h(-t+z+ x,m-p+ j;\omega_1,\omega_2)      \\
	&\times  e^{i\pi \big(\frac{(z+x)(-2t)}{\omega_1 \omega_2r} 
        + \frac{(m+j)(-2p)}{r}\big) } 
 \gamma_h(-t-z- x,-m-p- j;\omega_1,\omega_2) 
	 	\frac{d_j {x}}{2r\sqrt{-\omega_1\omega_2}}\:,
\end{aligned}\label{U1Moperator}
\end{equation}

where the spin-independent function is the same (\ref{spinindepinoperator}).

The following re-parametrization is needed to be able to use the integral identity (\ref{decorationequation2}) 

\begin{equation}
\begin{aligned}
     {a}_{1} & =  -s+ w, \quad {b}_{1}  = -s-w\:,
     \\
          {a}_{2} & =  -t+ x, \quad {b}_{2}  = -t-x\:,
     \\
    {u}_{1} & =  -q+ m, \quad {v}_{1}  =  -q- m
         \\
    {u}_{2} & =  -p+ k, \quad {v}_{2}  =  -p- k \:.\label{changeofvariables2}
\end{aligned}
\end{equation}

After the similar calculations presented in Appendix A in \cite{Gahramanov:2022jxz}, the decoration transformation (\ref{operatordecoration}) will be satisfied by the Bailey pair with $M$ operator (\ref{U1Moperator}).

\subsection{Flipping relation}

We will now be discussing Bailey pairs generated from an initial explicit pair. Noting that $M(t,p)_{z,m;x,j}$ is an integral-sum operator acting on a sequence of functions $f_j(x)$, the relation $\eqref{definition}$ suggests to start with $\alpha(x,j;t,p) = \delta_{jn}\delta(x-u)$ where $n\in \mathbb{Z}$, $u\in \mathbb{C}$ are new parameters.

Then, $\beta(z,m;t,p)$ of the following form
\begin{equation}
 \begin{aligned}
    \beta(z,m; t,p) &= M(t,p)_{z,m;x,j}\delta_{jn}\delta(x-u) \\
    & := M(t,p;z,m; u,n) \;,
\end{aligned}   
\end{equation}
forms a Bailey pair with $\alpha(x,j;t,p)$. From here, we generate new pairs with the Bailey lemma,
\begin{align}
    \alpha(x,k; t+s;p+q) &= I(s,q)\alpha(x,k;t,p)\:, \\
    \beta(x,k; t+s; p+q) &= I(-t,-p)M(s,q)_{x,k;z,m}I(s+t,p+q)\beta(z,m;t,p) \;. 
\end{align}

The relation $\eqref{definition}$ does not give us a particularly interesting result as it yields the star-triangle relation, which we have used to prove the Bailey lemma
\begin{equation}
\begin{aligned}\label{second bailey relation}
    M&(s,q)_{w,k;z,m}I(s+t, p+q)M(t,p;z,m;u,n) \\
    &= I(t,p)M(s+t, p+q ; w,k; u,n)I(s,q)  \;.
\end{aligned}
\end{equation}
An immediate consequence of $\eqref{second bailey relation}$ is the functions $\Tilde{\alpha}(z,m;s,q)$ and $\Tilde{\beta}(w,k;s,q)$ defined by
\begin{align}
\Tilde{\alpha}(z,m;s,q) &= I(s+t, p+q)M(t,p;z,m;u,n)\:, \\
\Tilde{\beta}(w,k;s,q) &=I(t,p)M(s+t, p+q ; w,k; u,n)I(s,q) \;,
\end{align}
form a Bailey pair with respect to parameters $s \in \mathbb{C}$, $q \in \mathbb{Z}$. Applying the lemma once again, we find
\begin{align}
    \Tilde{\alpha}'(z,m;s+c,q+d) =& I(c,d)I(s+t, p+q)M(t,p;z,m;u,n)\:, \\
    \begin{split}
    \Tilde{\beta}'(x,j;s+c,q+d) =& I(-s,-q)M(c,d)_{x,j;w,k}I(s+c,q+d)\\
    & \times I(t,p)M(s+t, p+q ; w,k; u,n)I(s,q)  \;, 
    \end{split}
\end{align}
where $a, c \in \mathbb{C}$ and $b,d \in \mathbb{Z}$ are arbitrary. The relation
\begin{equation}
\Tilde{\beta}'(x,j;s+c,q+d) = M(s+c,q+d)_{x,j;z,m} \Tilde{\alpha}'(z,m;s+c,q+d)   \;,
\end{equation}
yields a non-trivial integral identity
\begin{equation}
\begin{aligned}
&M(c,d)_{x,j;w,k}I(s+c,q+d)I(t,p)M(s+t, p+q ; w,k; u,n) \\
&= I(-s,-q)I(s,q) M(s+c,q+d)_{x,j;z,m} I(c,d)I(s+t, p+q)M(t,p;z,m;u,n) \;,
\end{aligned}
\end{equation}

but one sees the reduction 
\begin{equation}
\begin{aligned}
\label{operatorflipping}
M(c,d)_{x,j;w,k}M(s+t, p+q ; w,k; u,n) = M(s+c,q+d)_{x,j;z,m} M(t,p;z,m;u,n) \;,
\end{aligned}
\end{equation}
where the equality will be named "flipping relation" in our discussion.

\subsubsection{ \texorpdfstring{$SU(2)$}{SU(2)} gauge theory}

The integral identity (\ref{flippingequation1}) is already obtained double integral method \cite{Catak:2021coz} by applying the integral identity (\ref{decorationequation}) so it is expected to have the same integral-sum operator $M$. However, we need another operator due to the new definition of $\beta(z,m;t,p)$ after modifying  $\alpha(x,j;t,p) = \delta_{jn}\delta(x-u)$. Then the function $M(t,p;z,m;u,n)$ takes the following form

\begin{equation}
\begin{aligned}
M(t,p;z,m;u,n)=&
\frac{\Delta_n^u}{C(t,p)} \frac{1}{2r\sqrt{-\omega_1\omega_2}}
\gamma_h(-t\pm z\pm u,-p\pm m\pm n;\omega_1,\omega_2)  \:,
\end{aligned}\label{SU2M}
\end{equation}

where $C(t,p)$ is (\ref{spinindepinoperator}) and the measure (\ref{measure}) becomes 

\begin{equation}
\begin{aligned}
 \relax \Delta_n^u = \frac{ 1}{\gamma_h(\pm 2u, \pm   2n;\omega_1,\omega_2)}\:.
\end{aligned}
\end{equation}

During calculations, one can use the given change of variables

\begin{equation}
\begin{aligned}
     a_{1,2}&=-c\pm x, \quad u_{1,2}=-d\pm j\:,
     \\
          a_{3,4}&=-t-s\pm u, \quad u_{3,4}=-p-q \pm n\:,
     \\
    \Tilde{a}_{1,2}&=-c-s\pm x, \quad \Tilde{u}_{1,2}=-d-q\pm j
         \\
    \Tilde{a}_{3,4}&=-t\pm u, \quad \Tilde{u}_{3,4}=-p \pm n \:.\label{changeofvariables3}
\end{aligned}
\end{equation}

The integral sum operator $M$ (\ref{SU2Moperator}) and the operator (\ref{SU2M}) will satisfy the operator form of the flipping relation (\ref{operatorflipping}) with the integral identity (\ref{flippingequation1}) under the change of variables (\ref{changeofvariables3}).

\subsubsection{ \texorpdfstring{$U(1)$}{U(1)} gauge theory}

By the definition of the flipping relation (\ref{operatorflipping}), we introduce the following function as discussed also in the previous part

\begin{equation}
\begin{aligned}
M(t,p;z,m;u,n)=&
\frac{1}{C(t,p)}\frac{1}{r\sqrt{-\omega_1\omega_2}}
\gamma_h(-t+z+ u,m-p+ n;\omega_1,\omega_2)      \\
	&\times  e^{i\pi \big(\frac{(z+u)(-2t)}{\omega_1 \omega_2r} 
        + \frac{(m+n)(-2p)}{r}\big) } 
 \gamma_h(-t-z- u,-m-p- n;\omega_1,\omega_2) 
	 	\:,
\end{aligned}\label{U1M}
\end{equation}

where $C(t,p)$ is (\ref{spinindepinoperator}).

The following change of variables will be used to express the operators (\ref{U1Moperator}) and (\ref{U1M}) for Bailey pairs of (\ref{flippingequation2}) satisfy the flipping relation (\ref{operatorflipping})

\begin{equation}
\begin{aligned}
     {a}_{1} & =  -c+ w, \quad {b}_{1}  = -c-w\:, \quad
          {a}_{2}  =  -s-t+ x, \quad {b}_{2}  = -s-t-x\:,
     \\
    {u}_{1} & =  -d+ m, \quad {v}_{1}  =  -d- m\:, \quad
    {u}_{2}  =  -p-q+ k, \quad {v}_{2}  =  -p-q- k
    \\
         \Tilde{a}_{1} & =  -c-s+ x, \quad \Tilde{b}_{1}  = -c-s- x\:, \quad
          \Tilde{a}_{2} =  -t+ u, \quad \Tilde{b}_{2}  = -t-u\:,
     \\
    \Tilde{u}_{1} & =  -d-q+ j, \quad \Tilde{v}_{1}  =  -d-q- j\:, \quad
    \Tilde{u}_{2}  =  -p+ n, \quad \Tilde{v}_{2}  =  -p- n
    \:,\label{changeofvariables4}
\end{aligned}
\end{equation}

where parameters with a tilde are defined in (\ref{tilde2}).

\section{The relations for IRF-type models}

We will derive the decoration transformation and the flipping relation for IRF-type lattice spin models \cite{Date:1987wy, date:1987-2, BAXTER1986321}. Derivations will be based on the integrable edge interacting lattice spin models obtained by gauge/YBE correspondence. Therefore, we will use relations for Ising-tye models to describe IRF-type models. Similar constructions are studied to obtain solutions to IRF-type YBE by the use of the star-triangle relation and the star-star relation in which they are integrability conditions for nearest neighbor edge-interacting spin models. 

The Boltzmann weight of an interaction round a face of four spins will be factorized into four nearest-neighbor edge interactions by attaching a central spin with which all four spins interact

\begin{align}
R\left(\begin{array}{cc}
\sigma_{4} & \sigma_{3} \\
\sigma_{1} & \sigma_{2}
\end{array}\right)=\sum_{m_{i}} \int d x_{i} W\left(\sigma_{1}, \sigma_{i}\right) W\left(\sigma_{i}, \sigma_{2}\right) W\left(\sigma_{3}, \sigma_{i}\right) W\left(\sigma_{i}, \sigma_{4}\right) \:.
\label{IRFboltzmann}
\end{align}

The star-triangle relation and the decoration transformation for the same square lattice spin model are the tools to derive the decoration transformation and the flipping relation for IRF-type lattice spin models.  At first, we have a finite lattice Fig.\ref{squarefinite}. Then we apply the star-triangle relation and the decoration transformation to obtain figurative equalities.

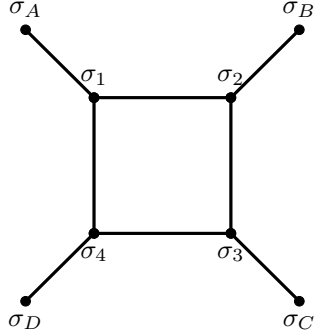
\begin{figure}[tbh]
\centering
\begin{tikzpicture}[scale=0.9]

\begin{scope}[xshift=-200pt]


\draw[-,very thick] (1,1)--(2,2);
\draw[-,very thick] (1,-1)--(2,-2);
\draw[-,very thick] (-1,-1)--(-2,-2);
\draw[-,very thick] (-1,1)--(-2,2);
\draw[-,very thick] (-1,-1)--(1,-1);
\draw[-,very thick] (-1,1)--(1,1);

\draw[-,very thick] (-1,1)--(-1,-1);
\draw[-,very thick] (1,1)--(1,-1);

\filldraw[fill=black,draw=black] (-2,-2) circle (2.0pt)
node[below=1pt]{\small $\sigma_D$};
\filldraw[fill=black,draw=black] (2,-2) circle (2.0pt)
node[below=1pt]{\small $\sigma_C$};
\filldraw[fill=black,draw=black] (2,2) circle (2.0pt)
node[above=1pt]{\small $\sigma_B$};
\filldraw[fill=black,draw=black] (-2,2) circle (2.0pt)
node[above=1pt]{\small $\sigma_A$};

\filldraw[fill=black,draw=black] (-1,-1) circle (2.0pt)
node[below=1pt]{\small $\sigma_4$};
\filldraw[fill=black,draw=black] (1,-1) circle (2.0pt)
node[below=1pt]{\small $\sigma_3$};
\filldraw[fill=black,draw=black] (1,1) circle (2.0pt)
node[above=1pt]{\small $\sigma_2$};
\filldraw[fill=black,draw=black] (-1,1) circle (2.0pt)
node[above=1pt]{\small $\sigma_1$};

\filldraw[fill=black,draw=black] (-1,-1) circle (2.0pt)
node[below=1pt]{\small };
\filldraw[fill=black,draw=black] (1,-1) circle (2.0pt)
node[below=1pt]{\small };
\filldraw[fill=black,draw=black] (1,1) circle (2.0pt)
node[above=1pt]{\small };
\filldraw[fill=black,draw=black] (-1,1) circle (2.0pt)
node[above=1pt]{\small };

\end{scope}

\end{tikzpicture}
\caption{Lattice for the derivation of the decoration transformation and the flipping relation for IRF-type models.}
\label{squarefinite}
\end{figure}

\subsection{The flipping relation}

The derivation of the flipping relation for IRF-type models is studied by the use of the star-triangle relation and the decoration relation of Ising-type models in a different order or for different sites. 

Applying the steps on Fig.\ref{squarefinite} to obtain the LHS  of Fig.\ref{IRFflippingfigure} are

\begin{itemize}
    \item the star-triangle relation on a star with central spin $\sigma_1$,
    \item the star-triangle relation for the triangle with  nodes $\sigma_2, \sigma_3, \sigma_4$,
    \item the decoration transformation on broken line with central spin $\sigma_3$,
    \item the star-triangle relation for the star with the external nodes $\sigma_2, \sigma_4, \sigma_C$,
    \item the decoration transformation for the nodes $\sigma_2, \sigma_4$.
\end{itemize}

The steps for the RHS of the flipping relation depicted in Fig.\ref{IRFflippingfigure} are

\begin{itemize}
    \item the star-triangle relation on a star with central spin $\sigma_2$,
    \item the star-triangle relation for the triangle with the nodes $\sigma_1, \sigma_3, \sigma_4$,
    \item the decoration transformation on a broken line with a centered node $\sigma_4$,
    \item the star-triangle relation for the star with surrounding nodes $\sigma_1, \sigma_3, \sigma_D$,
    \item the decoration transformation for the interaction between $\sigma_1, \sigma_3$.
\end{itemize}

After following the steps we reached pictorial equality of the flipping relation in Fig.\ref{IRFflippingfigure} which is the trivial version of the IRF-type YBE \cite{BAXTER1986321, date:1987-2}. We will see that it is actually the star-star relation. 

\begin{figure}[tbh]
\centering
\begin{tikzpicture}[scale=0.9]

\begin{scope}[xshift=-340pt]


\draw[-,very thick] (1,1)--(2,2);
\draw[-,very thick] (-1,-1)--(-2,-2);
\draw[-,very thick] (-1,-1)--(2,-2);
\draw[-,very thick] (-2,2)--(1,1);
\draw[-,very thick] (-2,2)--(-1,-1);
\draw[-,very thick] (1,1)--(2,-2);

\draw[-,very thick] (1,1)--(-1,-1);

\filldraw[fill=black,draw=black] (-2,-2) circle (2.0pt)
node[below=1pt]{\small $\sigma_D$};
\filldraw[fill=black,draw=black] (2,-2) circle (2.0pt)
node[below=1pt]{\small $\sigma_C$};
\filldraw[fill=black,draw=black] (2,2) circle (2.0pt)
node[above=1pt]{\small $\sigma_B$};
\filldraw[fill=black,draw=black] (-2,2) circle (2.0pt)
node[above=1pt]{\small $\sigma_A$};

\filldraw[fill=black,draw=black] (-1,-1) circle (2.0pt)
node[below=1pt]{\small $\sigma_4$};
\filldraw[fill=black,draw=black] (0,0) circle (2.0pt)
node[below=1pt]{\small $\sigma_0$};
\filldraw[fill=black,draw=black] (1,1) circle (2.0pt)
node[above=1pt]{\small $\sigma_2$};

\filldraw[fill=black,draw=black] (-1,-1) circle (2.0pt)
node[below=1pt]{\small };
\filldraw[fill=black,draw=black] (0,0) circle (2.0pt)
node[below=1pt]{\small };
\filldraw[fill=black,draw=black] (1,1) circle (2.0pt)
node[above=1pt]{\small };

\end{scope}

\fill[white!] (-8,0) circle (0.01pt)
node[left=0.05pt] {\color{black}$=$};

\begin{scope}[xshift=-120pt]


\draw[-,very thick] (-1,1)--(1,-1);
\draw[-,very thick] (1,-1)--(2,-2);
\draw[-,very thick] (-1,1)--(-2,2);
\draw[-,very thick] (-2,-2)--(1,-1);
\draw[-,very thick] (-1,1)--(2,2);

\draw[-,very thick] (-1,1)--(-2,-2);
\draw[-,very thick] (2,2)--(1,-1);

\filldraw[fill=black,draw=black] (0,0) circle (2.0pt)
node[below=1pt]{\small $\sigma_0$};
\filldraw[fill=black,draw=black] (1,-1) circle (2.0pt)
node[below=1pt]{\small $\sigma_3$};
\filldraw[fill=black,draw=black] (-1,1) circle (2.0pt)
node[above=1pt]{\small $\sigma_1$};

\filldraw[fill=black,draw=black] (-2,-2) circle (2.0pt)
node[below=1pt]{\small $\sigma_D$};
\filldraw[fill=black,draw=black] (2,-2) circle (2.0pt)
node[below=1pt]{\small $\sigma_C$};
\filldraw[fill=black,draw=black] (2,2) circle (2.0pt)
node[above=1pt]{\small $\sigma_B$};
\filldraw[fill=black,draw=black] (-2,2) circle (2.0pt)
node[above=1pt]{\small $\sigma_A$};

\filldraw[fill=black,draw=black] (1,-1) circle (2.0pt)
node[below=1pt]{\small };
\filldraw[fill=black,draw=black] (0,0) circle (2.0pt)
node[above=1pt]{\small };
\filldraw[fill=black,draw=black] (-1,1) circle (2.0pt)
node[above=1pt]{\small };

\end{scope}
\end{tikzpicture}
\caption{The flipping relation of IRF-type models.}
\label{IRFflippingfigure}
\end{figure}
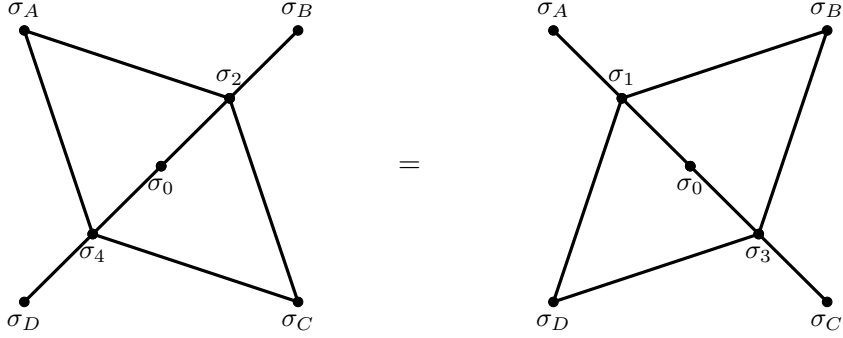

The flipping relation Fig.\ref{IRFflippingfigure} can be written in mathematical form

   \begin{align}
    \begin{aligned}
\sum_{m_0 } \int dx_0  \; R\left( \begin{array}{cc}
  {\sigma_A} & {\sigma_0}\\ {\sigma_B} & 
  {\sigma_C}\end{array}\right) \; R\left(\begin{array}{cc}
    {\sigma_0} & {\sigma_C}\\ {\sigma_A} &
    {\sigma_D}\end{array}\right)
 \makebox[8em]{}
      \\
\makebox[8em]{} 
= \sum_{m_0 } \int dx_0 \; R\left(\begin{array}{cc}
        {\sigma_D} & {\sigma_0}\\ {\sigma_B} &
        {\sigma_A}\end{array}\right)
 R\left(\begin{array}{cc}
          {\sigma_D} & {\sigma_B}\\ {\sigma_C} &
          {\sigma_0}\end{array}\right) \:.\label{IRFflippingequation}
\end{aligned}
\end{align}

where the definition of Boltzmann weight of IRF-type models (\ref{IRFboltzmann}) are used.

\subsection{The decoration transformation}

The idea is similar, we apply the star-triangle relation and the decoration transformation of the edge-interacting lattice spin model. To do so, we start with one side of the flipping relation which we will take RHS of Fig.\ref{flippingfigure}. The steps are the following

\begin{itemize}
    \item the decoration transformation for the interactions of $\sigma_0$,
    \item the star-triangle relation for the triangle with the nodes $\sigma_1, \sigma_3, \sigma_D$,
    \item the star-triangle relation on the star with central spin $\sigma_1$,
    \item the star-triangle relation for the triangle with nodes $\sigma_B, \sigma_3, \sigma_0$ where $\sigma_0$ appeared after second step,
    \item the decoration transformation on the broken line centered with the node $\sigma_3$,
    \item the star-triangle relation for the star with surrounding nodes $\sigma_B, \sigma_C, \sigma_0$.
\end{itemize}

Here we note that the procedure has a crucial point at the third step. If one applies the star-triangle relation for the star with the central node $\sigma_3$, the figurative result is the same but the external Boltzmann weights and two of the centered Boltzmann weights interchange. That is, the Bolztmann weights $W(\sigma_A,\sigma_0)$ and $W(\sigma_C,\sigma_0)$ take place at $W(\sigma_A,\sigma_B)$ and $W(\sigma_C,\sigma_B)$ respectively, and vice versa. 

\begin{figure}[tbh]
\centering
\begin{tikzpicture}[scale=0.9]

\begin{scope}[xshift=-320pt]


\draw[-,very thick] (-1,1)--(1,-1);
\draw[-,very thick] (1,-1)--(2,-2);
\draw[-,very thick] (-1,1)--(-2,2);
\draw[-,very thick] (-2,-2)--(1,-1);
\draw[-,very thick] (-1,1)--(2,2);

\draw[-,very thick] (-1,1)--(-2,-2);
\draw[-,very thick] (2,2)--(1,-1);

\filldraw[fill=black,draw=black] (0,0) circle (2.0pt)
node[below=1pt]{\small $\sigma_0$};
\filldraw[fill=black,draw=black] (1,-1) circle (2.0pt)
node[below=1pt]{\small $\sigma_3$};
\filldraw[fill=black,draw=black] (-1,1) circle (2.0pt)
node[above=1pt]{\small $\sigma_1$};

\filldraw[fill=black,draw=black] (-2,-2) circle (2.0pt)
node[below=1pt]{\small $\sigma_D$};
\filldraw[fill=black,draw=black] (2,-2) circle (2.0pt)
node[below=1pt]{\small $\sigma_C$};
\filldraw[fill=black,draw=black] (2,2) circle (2.0pt)
node[above=1pt]{\small $\sigma_B$};
\filldraw[fill=black,draw=black] (-2,2) circle (2.0pt)
node[above=1pt]{\small $\sigma_A$};

\filldraw[fill=black,draw=black] (1,-1) circle (2.0pt)
node[below=1pt]{\small };
\filldraw[fill=black,draw=black] (0,0) circle (2.0pt)
node[above=1pt]{\small };
\filldraw[fill=black,draw=black] (-1,1) circle (2.0pt)
node[above=1pt]{\small };

\end{scope}

\fill[white!] (-7,0) circle (0.01pt)
node[left=0.05pt] {\color{black}$=$};

\begin{scope}[xshift=-110pt]


\draw[-,very thick] (-2,2)--(2,2);
\draw[-,very thick] (2,-2)--(2,2);
\draw[-,very thick] (0,0)--(-2,2);
\draw[-,very thick] (2,-2)--(0,0);
\draw[-,very thick] (-2,-2)--(0,0);
\draw[-,very thick] (0,0)--(2,2);

\filldraw[fill=black,draw=black] (0,0) circle (2.0pt)
node[below=1pt]{\small $\sigma_0$};

\filldraw[fill=black,draw=black] (-2,-2) circle (2.0pt)
node[below=1pt]{\small $\sigma_D$};
\filldraw[fill=black,draw=black] (2,-2) circle (2.0pt)
node[below=1pt]{\small $\sigma_C$};
\filldraw[fill=black,draw=black] (2,2) circle (2.0pt)
node[above=1pt]{\small $\sigma_B$};
\filldraw[fill=black,draw=black] (-2,2) circle (2.0pt)
node[above=1pt]{\small $\sigma_A$};

\filldraw[fill=black,draw=black] (0,0) circle (2.0pt)
node[above=1pt]{\small };

\end{scope}
\label{IRFdecorationfigure}
\end{tikzpicture}
\caption{The decoration transformation for IRF-type models.}
\end{figure}
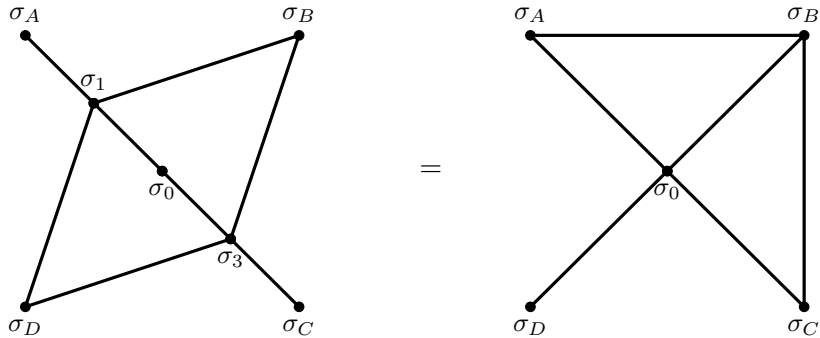

The mathematical expression of the Fig.\ref{IRFdecorationfigure} is written below and it is also been studied as the fusion procedure \cite{Date:1987wy, date:1987-2}

   \begin{align}
    \begin{aligned}
\sum_{m_0 } \int dx_0  \; R\left( \begin{array}{cc}
  {\sigma_A} & {\sigma_0}\\ {\sigma_B} & 
  {\sigma_D}\end{array}\right) \; R\left(\begin{array}{cc}
    {\sigma_0} & {\sigma_B}\\ {\sigma_D} &
    {\sigma_C}\end{array}\right)
=W(\sigma_A,\sigma_B)W(\sigma_C,\sigma_B)
R\left(\begin{array}{cc}
    \sigma_A    & \sigma_C \\
\sigma_B         & \sigma_D
    \end{array}\right)\:.\label{IRFdecorationequation}
\end{aligned}
\end{align}
where the remaining two edge interactions can be eliminated \cite{date:1987-2, Baxter:1997ssr}.

If one also applies decoration transformation for IRF-type models to LHS of Fig.\ref{IRFflippingfigure}, then expect to obtain the star-star relation \cite{Baxter:1997ssr}. However, the derivation of the star-star relation needs more attention and the discussion is left for further studies.

\section{Conclusions}\label{conclusion}

In this work, we constructed hyperbolic hypergeometric solutions to the decoration transformation and the flipping relation via the integral identities obtained by equality of partition functions of dual $\mathcal{N}=2$ supersymmetric gauge theories on three-dimensional squashed lens space $S^3_b/\mathbb{Z}_r$. We studied two solutions for both transformations and we obtained the latter solutions by gauge symmetry breaking method. 

The Bailey pairs for the decoration transformation and the flipping relation are constructed by the use of the $M$ operator which is used for the construction of the star-triangle relation and the star-star relation \cite{Gahramanov:2022jxz}. The decoration transformation and the flipping relation for IRF-type models \cite{BAXTER1986321, Date:1987wy, date:1987-2} are also derived.

There are various directions and further studies to be completed. In this study, the integral identities (\ref{decorationequation}) and (\ref{decorationequation2}) are studied as the decoration transformation but its $r=1$ case and its Euler gamma limit version are studied as the star-triangle relation in \cite{Kels:2018xge}. So it could be interesting to obtain a generalized version of Barnes’s first lemma \cite{Barnes1908} by the use of Euler gamma limit ($r \to \infty$) \cite{Eren:2019ibl} for (\ref{decorationequation2}). 
Also, basic hypergeometric beta integrals \cite{vandebult2013basic, Gahramanov:2016wxi} can be discussed as the solutions of the decoration transformation as well.
From the lattice spin models point of view, it is not clear to us how to study the transformation formula in Theorem 5.6.17 \cite{BultThesis} and the equation (7.26) in \cite{Spiridonov:2011hf}. Hence, the flipping relation may need more attention.
The intersections between the flipping relation and the commutativity of Baxter's Q-operators for Ruijsenaars-Sutherland hyperbolic systems \cite{Belousov:2023sat} might be interesting. 




\section*{Acknowledgements}

We are grateful to Ilmar Gahramanov for his valuable discussions and further comments to improve the work. We are also thankful to Osman Ergeç. Mustafa Mullahasanoglu also would like to thank Amihay Hanany and A. P. Isaev for their enlightening discussions. We are also supported by Istanbul Integrability and Stringy Topics Initiative (\href{https://istringy.org/}{istringy.org}) Erdal Catak and Mustafa Mullahasanoglu are supported by the 3501-TUBITAK Career Development Program under grant number 122F451.

Mustafa Mullahasanoglu would like to thank the organizers of the following events during which the work was carried out: the "Young Researchers Integrability School and Workshop", Durham, 17–21 July 2023; the “XII. International Symposium on Quantum Theory and Symmetries", Prague, 24–28 July 2023; the "RDP School and Workshop on Mathematical Physics", Yerevan, 19-24 August 2023.

\appendix

\section{Notations and hyperbolic hypergeometric gamma function}\label{appendix1}

We introduce the definitions and notations of the hyperbolic hypergeometric gamma function \cite{van2007hyperbolic, Andersen:2014aoa} that we mainly use and benefit from in this work
\begin{align}
	\gamma^{(2)}(z;\omega_{1},\omega_{2})=e^{\frac{\pi i}{2}B_{2,2}(z;\omega_{1},\omega_{2})}\frac{(e^{2\pi i\frac{z}{\omega_{2}}}\tilde{q};\tilde{q})_\infty}{(e^{2\pi i\frac{z}{\omega_{1}}};q)_\infty} \; ,
\end{align}
where $\tilde{q}=e^{2\pi i \omega_{1}/\omega_{2}}$, $q=e^{-2\pi i \omega_{2}/\omega_{1}}$ , the $q$-Pochhammer symbol is 
\begin{align}
    (z;q)_{\infty}=\prod_{i=0}^{\infty}(1-zq^i)\:,
\end{align}
and the second Bernoulli polynomial is 
\begin{equation}
 B_{2,2}(z;\omega_1,\omega_2)=\frac{z^2}{\omega_1\omega_2}-\frac{z}{\omega_1}-\frac{z}{\omega_2}+\frac{\omega_1}{6\omega_2}+\frac{\omega_2}{6\omega_1}+\frac{1}{2}\:,
\end{equation}
 with the complex variables $\omega_{1}$, $\omega_{2}$.

There are several integral representations for the hyperbolic hypergeometric gamma function, one of them is the following
\begin{align}
	\gamma^{(2)}(z;\omega_{1},\omega_{2})=\exp{\left(-\int_{0}^{\infty}\frac{dx}{x}\left[\frac{\sinh{x(2z-\omega_{1}-\omega_{2})}}{2\sinh{(x\omega_{1})}\sinh{(x\omega_{2})}}-\frac{2z-\omega_{1}-\omega_{2}}{2x\omega_{1}\omega_{2}}\right]\right)} \; ,
\end{align}
where $Re(\omega_{1}),Re(\omega_{2})>0$ and $Re(\omega_{1}+\omega_{2})>Re(z)>0$. 

Some characteristic  properties of the function are
	\begin{align}
	\text{Symmetry:} & ~~~ & \gamma^{(2)}(z;\omega_{1},\omega_{2})=\gamma^{(2)}(z;\omega_{2},\omega_{1})  \\
	\text{Reflection:} & ~~~ & \gamma^{(2)}(z;\omega_{1},\omega_{2})\gamma^{(2)}(\omega_{1}+\omega_{2}-z;\omega_{1},\omega_{2})=1 \\
	\text{Scaling:} & ~~~ & \gamma^{(2)}(z;\omega_{1},\omega_{2})=\gamma^{(2)}(u z;u\omega_{1},u\omega_{2})\\
	\text{Conjugation:} & ~~~ & \gamma^{(2)}(z;\omega_{1},\omega_{2})^{*}=\gamma^{(2)}(z^{*};\omega_{2}^{*},\omega_{1}^{*})
 \\
	\text{Difference equation:} & ~~~ & \frac{\gamma^{(2)}(z+\omega_1;\omega_1,\omega_2)}{\gamma^{(2)}(z;\omega_1,\omega_2)}=2\sin\left(\frac{\pi z}{\omega_2}\right),~(\omega_1\leftrightarrow\omega_2)
	\end{align}
The asymptotic behavior of the hyperbolic hypergeometric gamma function is used while calculating flavor reduction and the gauge symmetry breaking 
	\begin{align}
	\lim_{z\to\infty}e^{\frac{\pi i}{2}B_{2,2}(z;\omega_{1},\omega_{2})}\gamma^{(2)}(z;\omega_{1},\omega_{2})=1\: \: \text{for} \:  \: \arg{\omega_{2}+\pi}>\arg{z}>\arg{\omega_{1}}   \\
	\lim_{z\to\infty}e^{-\frac{\pi i}{2}B_{2,2}(z;\omega_{1},\omega_{2})}\gamma^{(2)}(z;\omega_{1},\omega_{2})=1 \: \: \text{for} \: \: \arg{\omega_{2}}>\arg{z}>\arg{\omega_{1}-\pi} \; ,\label{asymptotics}
	\end{align}
	where $\text{Im}(\frac{\omega_{1}}{\omega_{2}})>0$.
We will not utilize but there is also another asymptotic property of the hyperbolic gamma function which reduces to the Euler gamma function
\begin{equation}
	\lim_{\omega_2\to\infty} \Big(\frac{\omega_2}{2\pi\omega_1}\Big)^{\frac{z}{\omega_2}-\frac{1}{2}}\gamma^{(2)}(z;\omega_1,\omega_2)=\frac{\Gamma(z/\omega_1)}{\sqrt{2\pi}} .
	\label{gamma_limit}
	\end{equation}
 
However, we will introduce a lens version of the hyperbolic hypergeometric gamma function \cite{Gahramanov:2016ilb} which is related to the improved double sine function \cite{Nieri:2015yia} 
\begin{align}
\gamma_h(z,y;\omega_1,\omega_2) 
=    \gamma^{(2)}(-iz-i\omega_1y;-i\omega_1r,-i\omega) \: \gamma^{(2)}(-iz-i\omega_2(r-y);-i\omega_2r,-i\omega) \:,
\end{align}
where $\omega:=\omega_1+\omega_2$ and this shorthand notation is used in the rest of the paper.

The reflection property of the lens hyperbolic hypergeometric gamma function is
\begin{equation}
    \begin{aligned}
    \gamma_h( z, y;\omega_1,\omega_2)\gamma_h( \omega-z,r- y;\omega_1,\omega_2) =1
\:,
\end{aligned}
\end{equation}
or  
\begin{equation}
    \begin{aligned}
    \gamma_h( z, y;\omega_1,\omega_2)\gamma_h( \omega-z, y;\omega_2,\omega_1) =1
\:,
\end{aligned}
\end{equation}
where the symmetry and reflection properties of hyperbolic hypergeometric gamma function are mixed.

We also carry out the following shorthand notation
\begin{equation}
    \begin{aligned}
    \gamma_h(\pm z,\pm y;\omega_1,\omega_2) =\gamma_h( z, y;\omega_1,\omega_2)\gamma_h(- z,- y;\omega_1,\omega_2)\:.
\end{aligned}
\end{equation}

\bibliographystyle{utphys}
\bibliography{refYBE}

\end{document}